\documentclass[letterpaper, 10 pt, conference]{IEEEtran}

\IEEEoverridecommandlockouts                              

\usepackage{geometry} 
\usepackage[compact]{titlesec}
    \titlespacing{\section}{0pt}{1ex}{2ex}
\usepackage{amsfonts} 
\usepackage{paralist}
\usepackage{subcaption}  
\usepackage{amssymb}
\usepackage{amsfonts}
\usepackage{array,booktabs}
\usepackage{commath}
\usepackage{amsmath}
\usepackage{amsthm}
\usepackage{comment}
\usepackage{pdfcomment}
\usepackage[version=4]{mhchem}
\usepackage[commandnameprefix=always]{changes}
\usepackage{enumitem}
\usepackage{tikz}
\usepackage{hyperref}
\usepackage{cite}

\usepackage{cleveref}
\usepackage[utf8]{inputenc}
\usepackage[siunitx]{circuitikz} 
\usetikzlibrary{decorations,backgrounds}
\usepackage{chemformula} 
\usepackage[T1]{fontenc} 
\usepackage{subcaption,ragged2e}    
\usepackage{graphicx}
\usepackage{multicol}
\tikzstyle{block} = [rectangle, draw=violet, line width=1.25pt, fill=violet!10, text centered, rounded corners,inner sep=4pt,
    outer sep=1.5pt]
\usetikzlibrary{arrows}
\usetikzlibrary{shapes}
\usetikzlibrary{positioning}
\tikzset{
block2/.style={
  draw=blue!99, line width=1.25pt, fill=blue!10, text centered,
  rectangle, rounded corners,
  inner sep=4pt  ,  outer sep=1.5pt,
  },
sum/.style={
  draw=blue!90, line width=1.5pt, fill=violet!10, 
  circle, 
  },
input/.style={coordinate},
output/.style={coordinate},
header/.style={draw, text centered},
pinstyle/.style={coordinate}
} 
\DeclareMathOperator*{\argmin}{arg\,min}

\theoremstyle{definition}
\newtheorem{theorem}{Theorem}
\theoremstyle{definition}
\newtheorem{definition}{Definition}
\theoremstyle{definition}

\theoremstyle{definition}

\theoremstyle{definition}
\newtheorem{lemma}{Lemma}
\theoremstyle{definition}

\usepackage{pifont}
\newcommand{\cmark}{\ding{51}}%
\newcommand{\xmark}{\ding{55}}%

\usepackage[font={small}]{caption}
\graphicspath{{figures/}}
\usepackage{tabularx}

\title{\LARGE \bf
Nonlinear Optimal Control of DC Microgrids with \\
Safety and Stability Guarantees}

\author{Muratkhan Abdirash$^{1}$,~\IEEEmembership{Student Member,~IEEE}, Xiaofan Cui$^{1}$,~\IEEEmembership{Member,~IEEE}
\thanks{$^{1}$Muratkhan Abdirash and Xiaofan Cui are with the Department of Electrical and Computer Engineering, University of California, Los Angeles, Los Angeles, CA, 90095 {\tt\small mabdirash@ucla.edu, cuixf@seas.ucla.edu} (corresponding author).}
}

\begin{document}
\maketitle
\thispagestyle{empty}
\pagestyle{empty}

\begin{abstract}
A DC microgrid is a promising alternative to the traditional AC power grid, since it can efficiently integrate distributed and renewable energy resources. However, as an emerging framework, it lacks the rigorous theoretical guarantees of its AC counterpart. In particular, safe stabilization of the DC microgrid has been a non-trivial task in power electronics. To address that, we take a control theoretic perspective in designing the feedback controller with provable guarantees. We present a systematic way to construct Control Lyapunov Functions (CLF) to stabilize the microgrid, and, independently, Control Barrier Functions (CBF) to enforce its safe operation at all times.
The safety-critical controller (SCC) proposed in this work integrates the two control objectives, with safety prioritized, into a quadratic program (QP) as linear constraints, which allows for its 
online deployment using off-the-shelf convex optimization solvers. The SCC is compared against a robust version of the conventional droop control through numerical experiments whose results indicate the SCC outperforms the droop controller in guaranteeing safety and retaining stability at the same time.

\end{abstract}

\section*{Nomenclature}

\begin{tabular}{c|c}
DC &Direct current\\

AC&Alternating current\\

KCL &Kirchhoff's current law\\

KVL &Kirchhoff's voltage law\\

CPL &Constant power load\\

CCP& Continuous control property\\

CLF &Control Lyapunov function\\

CBF &Control barrier function\\

QP &Quadratic program\\

$v_j$ &Output voltage of $j$-th converter [V]\\

$i_{s_j}$ &Current supplied to $j$-th source converter [A]\\

$i_{t_j}$ &Current through $j$-th transmision line [A]\\

$v_L$ &Voltage across the aggregated load [V]\\

$C_j$ &Output capacitance of $j$-th converter [mF]\\

$R_j$ &Resistance of $j$-th transmission line [m$\Omega$]\\

$L_j$ &Inductance of $j$-th transmission line [mH]\\

$C_L$ &Load capacitance [mF]\\

$R_L$ &Load resistance [$\Omega$]\\

$P_L$ &Power rating of CPL [W]\\

$I_{\max}$ &Maximum rated current of CPL\\ 

$V_{\min}$ &Minimum rated voltage of CPL\\ 
$\mathbb{N}$ &Set of natural numbers\\
$\partial\mathcal{C}$& Boundary of a set $\mathcal{C}\subset\mathbb{R}^n$ \\
$\text{int}(\mathcal{C})$& Interior of  set $\mathcal{C}\subset\mathbb{R}^n$
\end{tabular}
\section{Introduction}
At the edge of traditional AC grids, advancements in power electronics have led to the emergence of DC microgrids, enabling more cost-effective and energy-efficient integration of DC loads across various sectors. The notable of which are telecom stations, data centers, electric vehicles, and zero-net energy buildings. 
Despite these benefits, a DC microgrid has presented a number of distinct challenges in the efforts of establishing its safe operation and stability. Those concerns include its ability to handle nonlinear loads, low system inertia, and rapid transients that are on the order of milliseconds, all induced by power electronics devices connecting the nodes of the DC microgrid \cite{shahgholian2021brief}.
\begin{figure}[t]
\centering
\includegraphics[width=0.45\textwidth,trim={7.5cm 1cm 6.5cm 0cm},clip]{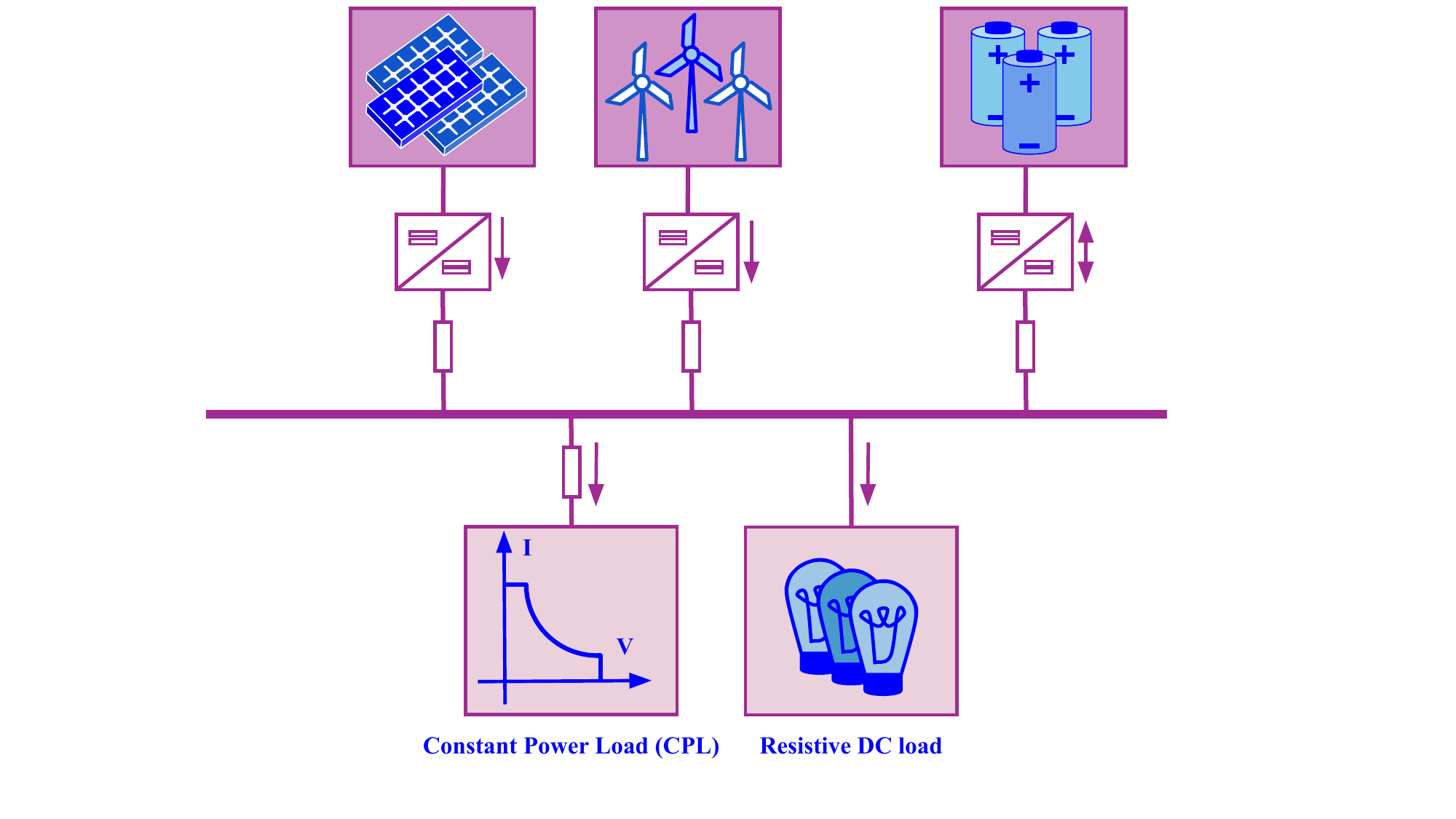} \caption{Illustrative example of a single-bus DC microgrid. It delivers the energy generated by DC power sources like solar panels and batteries to DC loads like electronic devices through one bus.}\label{fig:high-level}
\end{figure} 

\textit{DC/DC converters} serve as the critical power electronic interface in the DC microgrid, regulating voltage transitions between the source side, the bus, and the load.  Fig.\ref{fig:high-level} presents a typical architecture of a single-bus DC microgrid with DC/DC converters interlinking the nodes of the network. Because of its efficient nature, a DC/DC converter on the load side can be approximated as a constant power load (CPL) that keeps the power flowing in equal to the one flowing out \cite{chang2020large,chang2021region}.  On the flip side, CPL is non-linear and destabilzing, making it fairly easy to disturb microgrid's DC power supply system and violate the safe operating conditions of the entire network.

To address the destabilizing effects of constant power loads (CPLs), various control strategies have been developed, each offering different levels of theoretical guarantees. For instance, the controllers proposed in \cite{chang2020large,chang2021region} can provide large-signal stability, albeit contingent on the subset of model parameter values and/or restricted region of attraction. 
More computationally advanced model-based approaches, such as Model Predictive Control (MPC) \cite{zhang2022decentralized} and Active Disturbance Rejection Control (ADRC) \cite{li2024decentralized}, can additionally enforce safety constraints besides large-signal stability requirement. Their drawback comes in providing the provable guarantees of stability and safety for a bus with an arbitrary number of converters.
On the other hand, practical control strategies have been developed to enhance stability while prioritizing ease of implementation. A notable example is decentralized droop control with physical voltage, current, and power limiters \cite{10379140}. While this approach simplifies controller realization, it lacks large-signal stability guarantees, making it less robust against severe disturbances \cite{marimuthu2023review}.

Recent advances in control theory have led to the development of online controllers that solve constrained quadratic programs (QPs) to simultaneously enforce safety and stability constraints \cite{ames2016control}. Among these, the QP formulation proposed in \cite{jankovic} stands out due to its strong theoretical guarantees, ensuring safety constraint satisfaction at all times and achieving local asymptotic stability under mild conditions.
Building on this, we propose a centralized online controller for DC microgrids, where all converters are coordinated holistically as a single unit through a constrained QP framework. 


Contributions of this work, which also reveal the components of the proposed controller, are captured below: 
\begin{enumerate}
\vspace{-0.1cm}
    \item A Control Lyapunov Function (CLF) with large signal (exponential) stability guarantee for a single-bus DC microgrid; 
    \item First application, to our knowledge, of a Control Barrier Function (CBF), to enforce converter terminal voltage safety in an island DC microgrid;
    \item Online safety-critical controller (SCC) formulated as a solution to a QP, linearly constrained by CLF and CBF conditions. Any commercial convex optimization tool should be able to solve the QP in practice;
    \item Lastly, numerical experiments comparing performance of the SCC against a robust droop controller with large signal stability guarantee \cite{herrera2015stability}. The results show that the SCC can guarantee safety and local exponential stability from a larger set of initial conditions than the droop controller.
\end{enumerate}
   

\section{Preliminaries}
Consider a nonlinear \textit{control-affine} dynamical system
\begin{equation}
    \dot{x}=f(x)+g(x)u,\label{aff}
\end{equation}
where its state and control input are given by $x\in X\subseteq\mathbb{R}^n$ and $u\in U\subseteq\mathbb{R}^m$ ($x_j$ and $u_l$ refer to $j$-th entry of $x$ and $l$-th entry of $u$, respectively). Note that $X$ is called the state-space, and $U$ is the admissible set of controls, while $f:X\rightarrow\mathbb{R}^n$ is the drift vector, and $g:X\rightarrow\mathbb{R}^{n\times m}$ is the control matrix. We will now go over important notions and tools from control theory.
\begin{definition}\label{equil}
    A pair $(u^*,x^*)\in U\times X$ is an \textit{equilibrium} of system \eqref{aff} if
    \begin{equation}
        f(x^*)+g(x^*)u^*=0_{n}.
    \end{equation}
\end{definition}
\begin{definition}\label{lie_deri}
     A \textit{Lie derivative}, $L_f(V)$, is an operator that differentiates a function $V:\mathbb{R}^n\rightarrow\mathbb{R}$ along a vector field $f:\mathbb{R}^n\rightarrow\mathbb{R}^n$, which generalizes the directional derivative. The Lie derivative can be applied consecutively via $L_f^\gamma(\cdot)=\frac{\partial L_f^{\gamma-1}(\cdot)}{\partial x}f(x)$ where $L_f^1(\cdot)=L_f(\cdot)=\frac{\partial(\cdot)}{\partial x}f(x)$. Note that the gradient operator $\frac{\partial(\cdot)}{\partial x}$ produces a \textit{row vector} valued function.
\end{definition}
\begin{definition}\label{lipschitz}
    Let $V:X\subset\mathbb{R}^n\rightarrow\mathbb{R}^n$ , then $V$ is \textit{Lipschitz continuous} on the its domain $X$ if there exists $M>0$ such that
    \begin{equation*}
        ||V(x)-V(y)||\leq M||x-y||,\forall x,y\in X,
    \end{equation*}
    where $||\cdot||$ is the Euclidean norm.
\end{definition}
\begin{definition}\label{rel_deg}
    Let $h:X\rightarrow\mathbb{R}$ be an output function for the dynamic system in \eqref{aff}. Then $h(x)$ is said to have a \textit{relative degree} of $\gamma\in\mathbb{N}$, if $\gamma\geq2$ is the smallest positive integer such that
    \begin{align*}
        &L_gL^j_fh(x)=0,\forall x\in X, j=1,\dots,\gamma-2, \\
        &||L_gL^{\gamma-1}_fh(x)||>\delta, \delta>0,\forall x\in X.
    \end{align*}
\end{definition}
\begin{definition}\label{glo_sta}
    Let $k(x)$ be a Lipschitz continuous controller for \eqref{aff} such that the closed-loop system is given by 
    \begin{equation}
        \dot{x}=f(x)+g(x)k(x),\label{closedloop}
    \end{equation}
    with $k(x^*)=u^*$. Let $x(t)$ denote the solution of \eqref{closedloop} to $x_0$. Then, closed-loop system \eqref{closedloop} is {\textit{exponentially stable}} at $x^*$, if there exist $M,\lambda>0$ such that
    \begin{equation*}
        x_0\in X\Rightarrow||x(t)-x^*||\leq M||x_0-x^*||e^{-\lambda t},\forall t>0.
    \end{equation*}
\end{definition}

\begin{definition}(Simplified)\label{CLF}
    Let $V:\mathbb{R}^n\rightarrow\mathbb{R^+}$ be a differentiable, radially unbounded, and positive definite function. Then, $V(x)$ is a \textit{Control Lyapunov function (CLF)} for \eqref{aff}, if there exists $\alpha>0$ such that
    \begin{equation}
        L_gV(x)=0\Rightarrow L_fV(x)+\alpha||x||^2<0,\forall x\neq0.
    \end{equation}
\end{definition}
\begin{definition}\label{ccp}
    A CLF $V$ satisfies the \textit{continuous control property} if for any $\epsilon>0$, there exists an open set $E\subseteq X$ such that for any $x\neq x^*\in E$, there exists a $u\in\mathbb{R}^m$ satisfying $||u-u^*||<\epsilon$ and $\dot{V}(x,u)<-\alpha||x-x^*||^2$ for $\alpha>0$.
\end{definition}
\begin{definition}\label{safeee} Let $\mathcal{C}\subset X$ denote the safe region of the state-space. A safe controller renders $\mathcal{C}$ \textit{forward invariant}, i.e. $x_0\in\mathcal{C}\Rightarrow x(t)\in\mathcal{C}, \forall t>0.$ 
\end{definition}
\begin{definition}(Simplified)\label{CBF}
    Let $B:\mathbb{R}^n\rightarrow\mathbb{R}$ be differentiable in the interior of the safe set $\text{int}(\mathcal{C})$. Suppose $B(x)$ grows unbounded as $x$ approaches the boundary of the safe set $\partial\mathcal{C}$. $B(x)$ is a \textit{control barrier function (CBF)} for \eqref{aff} if there exists $\beta>0$ such that
    \begin{equation}
        L_gB(x)=0\Rightarrow L_fB(x)-\frac{\beta}{B(x)}<0,\forall x\in\text{int}(\mathcal{C}).\label{cbfcon}
    \end{equation}
\end{definition}

Lastly, we briefly describe the current-voltage characteristics (the \textit{IV curve}) of a CPL. Graphically, it resembles a hyperbola truncated at both ends. However, we limit ourselves to a simplified model of the CPL \begin{equation}\begin{cases} 
        & I_{PL}=I_{\max}, V_L\leq V_{\min},\\
        & V_L=P_L/I_{PL}, V_L\geq V_{\min}.
    \end{cases}\label{cpl}\end{equation} where $P_L$ is CPL's power level kept constant, and $I_{PL}$ and $V_L$ are the current flowing through and the voltage across of it, respectively.   

\section{Dynamics of DC microgrid}
A single-bus DC microgrid can be described as a nonlinear electrical circuit, as shown in Fig. \ref{fig:generic}, where each DC/DC converter, along with its local current control loop, is simplified as a controlled current source.
The dynamics of the microgrid in Fig.\,\ref{fig:generic} are derived by applying KCL and KVL as follows
\begin{align}
     \underbrace{\dot{\begin{bmatrix}
    v_{1}\\
    i_{t_1}\\
    \vdots\\
    v_n\\
    i_{t_n}\\
    v_L
\end{bmatrix}}}_{:=\dot{x}}&=\underbrace{\begin{bmatrix}
    -i_{t_1}/C_{1}\\
    (v_{1}-i_{t_1}R_{1}-v_L)/L_{1}\\
    \vdots\\
    -i_{t_n}/C_{n}\\
    (v_{n}-i_{t_n}R_{n}-v_L)/L_{n}\\
    (i_{t_1}+\dots+i_{t_n}-v_L/R_L-I_{PL})/{C_L}
    \end{bmatrix}}_{:=f(x)}\nonumber\\&+\underbrace{\begin{bmatrix}
    1/{C_{1}}&0&\dots& 0\\
    0&0&\dots&0\\
     \vdots & \vdots & \ddots&\vdots\\
    0&0& \dots&1/C_n\\
    0&0&\dots&0\\
    0&0&\dots&0
    \end{bmatrix}}_{:=g(x)}\underbrace{\begin{bmatrix}
        i_{s_1}\\i_{s_1}\\\vdots\\i_{s_n}
    \end{bmatrix}}_{:=u}.
    \label{original}
\end{align}
The states are given by $j$-th converter's terminal voltage $v_j$ $j=1,\dots,n$, current through $j$-th transmission line $i_{t_j}$ with $j=1,\dots,n$, and bus voltage $v_L$. The control input is given by $i_{s_j}$. We collect all the states in $x\in\mathbb{R}^{2n+1}$ and control inputs in $u\in\mathbb{R}^n$, and denote the drift dynamics with $f:\mathbb{R}^{2n+1}\rightarrow\mathbb{R}^{2n+1}$ and the control matrix with $g:\mathbb{R}^{2n+1}\rightarrow\mathbb{R}^n$. Note that \eqref{original} is a control-affine system \eqref{aff}.

\textbf{Existence and Uniqueness of Solutions}: Under the simplified CPL model in \eqref{cpl}, $f$ and $g$ are both continuously differentiable. Given that the applied controller is also Lipschitz continuous, all the sufficient conditions required for the existense and uniqueness of solutions to \eqref{original} have been satisfied.

\textbf{Equilibrium:} {Finding the 
equilibrium of DC microgrid \eqref{original} results in an underdetermined nonlinear system of equations, so there are infinitely many equilibria of the microgrid, parameterized by the current distribution and desired bus voltage.} To address that, a convex optimization problem is formulated to choose a unique equilibrium. We minimize the steady-state power loss under the desired bus voltage $x^*$:
    \begin{align} 
\min_{x_{2}^*,x_{4}^*,\dots,x_{2n}^*}&\hspace{1cm} (x_{2}^*)^2R_{1}+(x_{4}^*)^2R_{2}+\dots+(x_{2n}^*)^2R_{n}\nonumber\\
        &\text{s.t. } \displaystyle x_{2}^*+x_{4}^*+\dots+x^*_{2n}=\frac{x_{2n+1}^*}{R_L}+\frac{P_L}{x_{2n+1}^*}.\label{minequil}
    \end{align} 
    The solution to \eqref{minequil} states the contribution of a particular power source in the overall supply to the load is inversely proportional to the resistance of the transmission line connecting the source and the load.  For the desired steady-state bus voltage $x_{2n+1}^*$, the equilibrium is given by the solution to \eqref{minequil}\begin{align}  x_1^*&=x_3^*=\dots=x_{2n-1}^*\nonumber\\&=\frac{\prod_{j=1}^nR_{j}}{\sum_{j=1}^n\prod_{i=1,i\neq j}^nR_{i}}\left(\frac{x_{2n+1}^*}{R_L}+\frac{P_L}{x_{2n+1}^*}\right)+x_{2n+1}^*;\vspace{1cm} \nonumber\\ x_2^*R_{1}&=x_4^*R_{2}=\dots=x^*_{2n}R_{n}\nonumber\\&=\frac{\prod_{j=1}^nR_{j}}{\sum_{j=1}^n\prod_{i=1,i\neq j}^nR_{i}}\left(\frac{x_{2n+1}^*}{R_L}+\frac{P_L}{x_{2n+1}^*}\right);\nonumber\\ \vspace{1cm}
    u_j^*&=x_{2j}^*,j=1,\dots,n.\label{u^*}
    \end{align}

\vspace{-20pt}
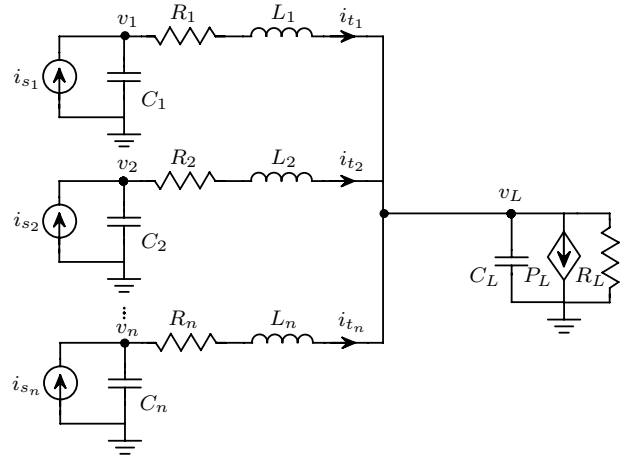
\begin{figure}[h]
\centering
\begin{tikzpicture} 
\node[inner sep=0,outer sep=0] at (0,0){\includegraphics[width=0.45\textwidth,trim={2cm 2cm 1cm 10cm},clip]
{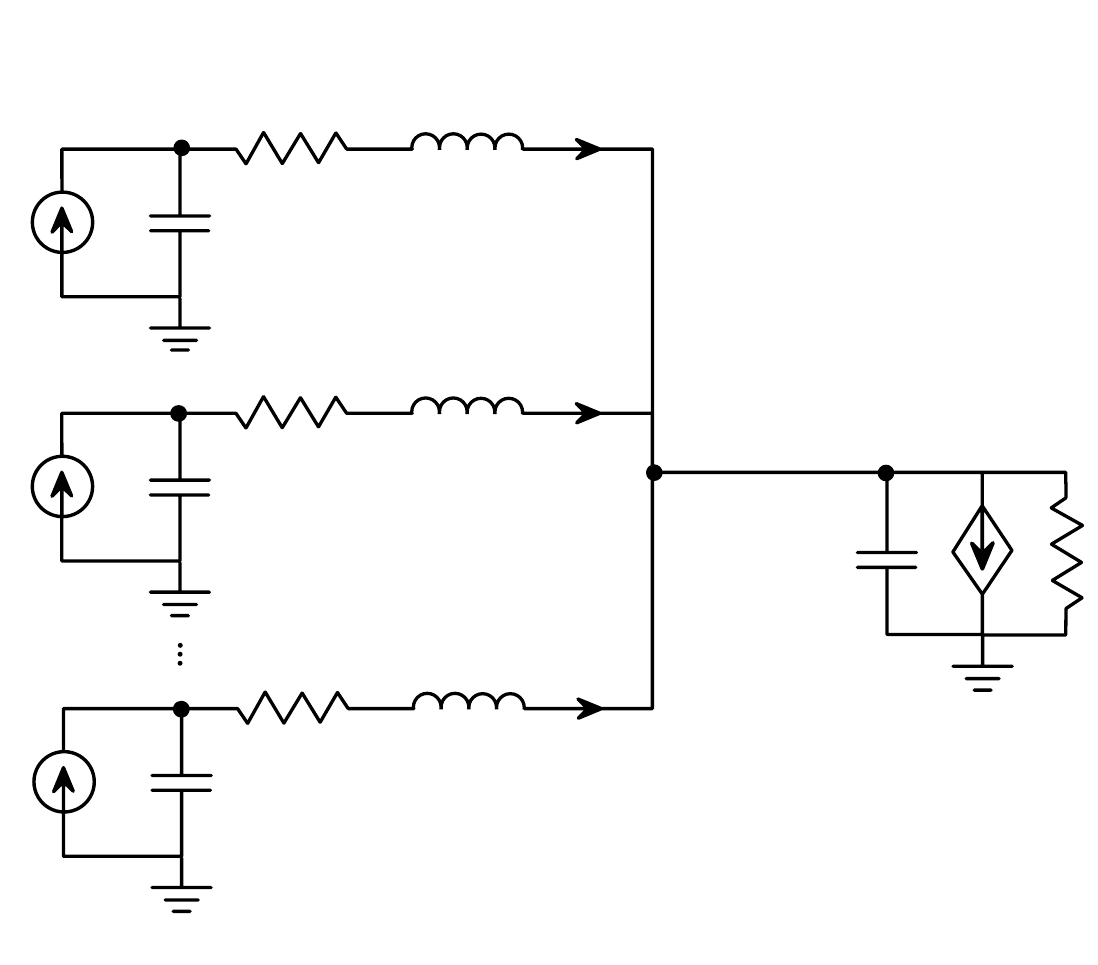}};
 \node[ ](ppp) at (-16pt,-38pt){\mbox{\footnotesize $L_n$}}; 
 \node[ ](ppp) at (-16pt,23pt){\mbox{\footnotesize $L_2$}};
 \node[ ](ppp) at (-16pt, 78pt){\mbox{\footnotesize $L_1$}};

  \node[ ](ppp) at (10pt,-39pt){\mbox{\footnotesize $i_{t_n}$}}; 
 \node[ ](ppp) at (10pt,22pt){\mbox{\footnotesize $i_{t_2}$}};
 \node[ ](ppp) at (10pt, 76pt){\mbox{\footnotesize $i_{t_1}$}};

   \node[ ](ppp) at (-75pt,-42pt){\mbox{\footnotesize $v_n$}}; 
 \node[ ](ppp) at (-75pt,20pt){\mbox{\footnotesize $v_2$}};
 \node[ ](ppp) at (-75pt, 75pt){\mbox{\footnotesize $v_1$}};

  \node[ ](ppp) at (-54pt,-38pt){\mbox{\footnotesize $R_n$}}; 
 \node[ ](ppp) at (-54pt,23pt){\mbox{\footnotesize $R_2$}};
 \node[ ](ppp) at (-54pt, 78pt){\mbox{\footnotesize $R_1$}};

  \node[ ](ppp) at (-65pt,-70pt){\mbox{\footnotesize $C_n$}}; 
 \node[ ](ppp) at (-65pt,-9pt){\mbox{\footnotesize $C_2$}};
 \node[ ](ppp) at (-65pt, 45pt){\mbox{\footnotesize $C_1$}};

  \node[ ](ppp) at (-113pt,-63pt){\mbox{\footnotesize $i_{s_n}$}}; 
 \node[ ](ppp) at (-113pt,-1pt){\mbox{\footnotesize $i_{s_2}$}};
 \node[ ](ppp) at (-113pt, 53pt){\mbox{\footnotesize $i_{s_1}$}};

  \node[ ](ppp) at (60pt,-22pt){\mbox{\footnotesize $C_L$}};
   \node[ ](ppp) at (80pt,-22pt){\mbox{\footnotesize $P_L$}};
    \node[ ](ppp) at (100pt,-22pt){\mbox{\footnotesize $R_L$}};
    \node[ ](ppp) at (69.9pt,9pt){\mbox{\footnotesize $v_L$}};
\end{tikzpicture}
\caption{Circuit schematics of a single-bus DC microgrid. $C_j$ is the output capacitance of $j$-th converter for $j=1,\dots,n$, while $C_L$ is the bus capacitance. Impedance of a transmission line is represented by its series resistance $R_{j}$ and inductance $L_{j}$ for $j=1,\dots,n$. The load is made up of an aggregated resistive load $R_L$ in parallel with a CPL with power level $P_L$.}\label{fig:generic} 
\end{figure}

\section{Large Signal Stable Controller Design}
\begin{definition}\label{QP_def}
   A \textit{quadratic program (QP)}  is a convex optimization problem with a quadratic cost function and linear constraints. It can be efficiently solved with any convex optimization tool.
\end{definition}
We now rewrite Theorem 1 \cite{jankovic} using the extension for a forced, i.e. $u^*\neq 0$, (rather than a natural one with $u^*=0$) equilibrium developed by \cite{taylor2023robust}.
\begin{theorem}\label{jankovic}
Let $V (x)$ be a CLF and $B(x)$ be a CBF for an open-loop system in \eqref{aff}. A feedback controller for the closed-loop system \eqref{closedloop}, $u:X\rightarrow U$ can be obtained by solving the following QP (Def \ref{QP_def})
\begin{align}
        &\min_{(u, \delta)\in U\times\mathbb{R}^{n}}||u-u_{FL}(x)||^2+m||\delta||^2\nonumber\\
        &\displaystyle\text{subject to }\gamma(L_{f_\eta}V_\eta(x)+\alpha||\eta(x)||^2)\nonumber\\&\hspace{2cm}+L_{g_\eta}V_\eta(x)\left(u+\delta\right) \leq\ 0\nonumber,\\
        &\hspace{1.5cm}\displaystyle L_fB(x)+L_gB(x)u-\frac{\beta}{B(x)}\leq0.
        \label{qp}
\end{align} 
where $\gamma(p):=\begin{cases}
    \frac{m+1}{m}p, p\geq0\\
     p, p<0,
\end{cases}$, $m$ is the penalty for a slack variable $\delta$, and  $u_{FL}(x)$ is a stabilizing controller obtained through feedback linearization. Then,
\begin{enumerate}
    \item The QP problem is feasible for all $x\in X$ and the resulting control law is Lipschitz continuous in every subset of $\text{int}(\mathcal{C})$ not containing the equilibrium.
    \item $\dot{B}(x)=L_fB(x)+L_gB(x)u\leq\frac{\beta}{B(x)}$ for all $x\in\text{int}(\mathcal{C})$ and the set $\text{int}(\mathcal{C})$ is forward invariant.
    \item If the barrier constraint is inactive and if we select $\gamma m = 1$ the control law achieves $\dot{V}_\eta(x)=L_{f_\eta}V_\eta(x)+L_{g_\eta}V_\eta(x)u\leq-\alpha||\eta(x)||^2.$
    \item If $x^*\in\text{int}(\mathcal{C})$ and the CLF $V_\eta$ has the CCP, then the barrier constraint is inactive around $x^*$, the control law is continuous at $x^*$, and the closed loop system is locally exponentially stable.
\end{enumerate}
\end{theorem}
\begin{proof}
    Please refer to the proof of Theorem 7 \cite{taylor2023robust}. The particular choice of $m$ is up to a designer, and Section 3 \cite{taylor2023robust} details the selection guidelines. Note that the resulting control law is Lipschitz continuous of order $m$. 
\end{proof}
\begin{definition}\label{feed_lin_def}
    Let $(x^*,u^*)$ be an equilibrium pair for an open-loop system \eqref{aff}. The system is \textit{feedback linearizable} with respect to the equilibrium, if there exist 
    \begin{itemize}
        \item a diffeomorphism $\Phi:X\rightarrow\mathbb{R}^n$ between $x\in X$ and $(\eta,z)\in\Phi(X)$ such that $\Phi(x)=(\eta(x),z(x))$ and $\Phi(x^*)=0_n$,
        \item constant $r\in\mathbb{N}$ with $r\leq n$,
        \item a controller $k_{fbl}:X\times\mathbb{R}^m\rightarrow\mathbb{R}^m$ that is Lipschitz continuous on $X\times\mathbb{R}^m$ and $k_{fbl}(x^*,0_m)=u^*$,
        \item a controllable pair $(F,G)\in\mathbb{R}^{r\times r}\times\mathbb{R}^{r\times k}$,
        \item $f_\eta:X\rightarrow\mathbb{R}^r$, $g_\eta:X\rightarrow\mathbb{R}^{r\times m}$, $z:X\rightarrow\mathbb{R}^{n-r}$ that are continuously differentiable
    \end{itemize}
    such that for all $x\in X$ the following hold
    \begin{enumerate}
        \item \begin{equation}
            \begin{bmatrix}
                f_\eta(x)\\\dot{z}(x)
            \end{bmatrix}=
                \frac{\partial\Phi}{\partial x}(x)f(x),\begin{bmatrix}
                    g_\eta(x)\\0_{n-r}
                \end{bmatrix}=\frac{\partial\Phi}{\partial x}(x)g(x);
        \end{equation}
        \item \begin{equation}
            f_\eta(x)+g_\eta(x)k_{fbl}(x,v)=F\eta(x)+Gv.
        \end{equation}
    \end{enumerate}
\end{definition}
    We embed our control objective of regulating the states to their desired steady-state values into the definition of the \textit{output functions}: 
    \begin{align}
        &h_0(x):=x_{2n+1}-x_{2n+1}^*,\nonumber\\
        &h_j(x):=x_{2j-1}-x_{2j+1}, \displaystyle j=1,\dots,\displaystyle n-1.\label{hj}
    \end{align}
    Note that $h_0$ measures how far the bus voltage is from its desired steady-state value, while $h_j$ measures the difference between the terminal voltages of converters $j$ and $j+1$. They are both motivated by \eqref{u^*}. 

\begin{definition}\label{out_def}
Recall Def. \ref{lie_deri} and consider the evolution of our objective functions, $\{h_j\}_{j=0}^{n-1}$ along the solutions of \eqref{original} expressed using Lie derivatives, then the \textit{output dynamics} is given by 
\begin{equation}
    \underbrace{\Dot{\begin{bmatrix}
    h_0(x)\\
    L_fh_0(x)\\
    L^2_fh_0(x)\\
    h_1(x)\\
    \vdots\\
    h_{n-1}(x)
\end{bmatrix}}}_{:=\dot{\eta}(x)}=\underbrace{\begin{bmatrix}
    L_fh_0(x)\\
    L^2_fh_0(x)\\
    L^3_fh_0(x)\\
    L_fh_1(x)\\
    \vdots\\
    L_fh_{n-1}(x)
    \end{bmatrix}}_{:=f_\eta(x)}+\underbrace{\begin{bmatrix}
   0\\
    0\\
    L_gL^2_fh_0(x)\\
    L_gh_1(x)\\
    \vdots\\
    L_gh_{n-1}(x)
    \end{bmatrix}}_{:=g_\eta(x)}u
    \label{output},
\end{equation}
 where the output state is $\eta\in\mathbb{R}^{n+2}$ with drift $f_\eta:X\rightarrow\mathbb{R}^{n+2}$ and control matrix $g_\eta:X\rightarrow\mathbb{R}^{n+2}$. Note that all the functions involved are continuously differentiable, so a Lie derivative can be applied consecutively. Recall Def. \ref{rel_deg},
and, note that $h_0$ has a relative degree of three, while $h_j$ has a relative degree of one for all $j=1,...,n-1.$ 
\end{definition}
\begin{definition}\label{zero_def}
    Let 
     \begin{equation}
         z_j(x):=R_{j}x_{2j}-R_{j+1}x_{2j+2}, j=1,\dots,n-1,\label{zero}
     \end{equation}
     be a function measuring the current sharing between input sources for \eqref{original}, motivated by the optimal steady-state values found in  \eqref{u^*}.
      Then, $z_j's$ are stacked onto $z\in\mathbb{R}^{n-1}$ and together with its evolution
      \begin{equation}
          \dot{z}(x)=L_fz(x),\label{zero_dyn}
      \end{equation} are referred to as \textit{zero dynamics}. 
\end{definition}

\begin{lemma}\label{lemma1}
    The output dynamics \eqref{output} together with the zero dynamics \eqref{zero_dyn} satisfy the conditions for feedback linearization given in Def. \ref{feed_lin_def}, and a CLF for the output dynamics is a CLF for the original system \eqref{original}:
    \begin{equation}\vspace{0.1cm}
        V_\eta(\eta)=\eta^TP\eta, \label{clf_eta_clf}
    \end{equation}
     with $P^T=P\in\mathbb{R}^{(n+2)\times (n+2)}$ positive definite being the solution to the Lyapunov equation  $A_{cl}^TP+PA_{cl}=-Q$ for some positive definite $Q^T=Q\in\mathbb{R}^{(n+2)\times (n+2)}$ and control matrix $K\in\mathbb{R}^{n\times(n+2)}$ of our choice such that
     \begin{equation}
         A_{cl}:=F+GK\label{acl}
     \end{equation}
     has eigenvalues with negative real parts.
\end{lemma}
\begin{proof}
Let $\tilde{g}_\eta\in\mathbb{R}^{n\times n}$ denote the control matrix of \eqref{output} with its first two identically zero rows removed. Similar, define $\tilde{f}_\eta\in\mathbb{R}^{n\times n}$ to be the drift vector from \eqref{output} with its first two entries removed. Importantly, $\tilde{g}_\eta$ comes out as an invertible matrix for all values of model parameters. Thus, we can apply a modified control input, $u=\tilde{g}_\eta^{-1}(-\tilde{f}_\eta+v)$ to cancel out the nonlinearities and arrive at
\begin{equation}
    \dot{\eta}=\underbrace{\begin{bmatrix}
            0_{n+1,1}&I_{n+1}\\0& 0_{n+1,1}^T
        \end{bmatrix}}_{:=F}\eta+\underbrace{\begin{bmatrix}0_{2,n}\\I_{n}
        \end{bmatrix}}_{:=G}v\hspace{-0.2cm}\label{linear}
\end{equation}
where $v$ is now the new control input. As $(F,G)$ is a controllable pair, there exist infinitely many $K$ such that the closed loop matrix $A_{cl}$ has eigenvalues with negative real parts. Then, we can pick such $K$, and embed it in the feedback linearizing controller to be applied to the output dynamics
\begin{equation}
    u_{FL}(x)= \tilde{g}_\eta^{-1}(-\tilde{f}_\eta+K\eta)\label{ufl}
\end{equation}Hence, the resulting output dynamics $\dot{\eta}=A_{cl}\eta$ is exponentially stable around $\eta^*=0_n.$ Once $\eta$ has been stabilized to $\eta^*$, we can show that $z_j$ evolves as a stable linear system with an eigenvalue of $\scriptstyle-\frac{(R_{j}+R_{j+1})}{(L_{j}+L_{j+1})}.$ Crucially, the coordinate transform, $\Phi$, from the original system \eqref{original} to the output+zero dynamics \eqref{output} such that $(\eta,z)=\Phi(x)$ is a diffeomorphism, since the Jacobian, $\frac{\partial \Phi}{\partial x}$, is full rank, implying  invertibility. 

    Recall the output dynamics was exponentially stable under the appropriate choice of the feedback gain $K$. The converse Lyapunov theory implies the existence of a quadratic Lyapunov function, $V_\eta(\eta)=\eta^TP\eta$, with $P^T=P$ being a positive definite matrix of dimension $n+2$. $P$ is obtained by solving the Lyapunov equation: $A_{cl}^TP+PA_{cl}=-Q$ for some positive definite $Q^T=Q$ of our choice. 

    The zero dynamics is a stable autonomous system, so its quadratic norm $V_z(z)=\frac{1}{2}||z||^2$ is a valid Lyapunov function. A weighted average of $V_\eta$ and $V_z$ can serve as a Lyapunov function to the original system \eqref{original} certifying its stability: 
$V(x)=\sigma V_\eta(\eta(x))+V_z(z(x))$ 
is just that for $\sigma$ large enough \cite{freeman2008robust}.
     Most notably, $V_\eta$ can effectively be used as a CLF to find more efficient controllers. In fact, any controller $u$ satisfying \begin{equation}
         L_fV_\eta+L_gV_\eta u\leq -\lambda_{\min}(\bar{Q})||\eta||^2,\label{clf_eta}
     \end{equation} will guarantee large signal stability of the closed-loop dynamics \eqref{closedloop} while maintaining Lipschitz continuity of $u$.
\end{proof}

\definecolor{amethyst}{rgb}{0.6, 0.4, 0.8}
\definecolor{blue(pigment)}{rgb}{0.2, 0.2, 0.6}
\definecolor{blue(pigment2)}{rgb}{0.2, 0.2, 0.4}
\begin{figure}[t]
        \centering
        \scriptsize
        \begin{tikzpicture}[auto,>=latex']
        \node[input]at (0pt,0pt)(input){};
            \node [block,left = 0pt of input](original) {$\color{violet}\dot{x}=f(x)+g(x)u$} ;

            \node[above = 0.01\linewidth of original]{\textcolor{violet}{DC microgrid dynamics}};
            \node [block, right = 0.225\linewidth of original] (output) {$\color{violet}\begin{bmatrix}
                 \dot{\eta}\\\dot{z}
             \end{bmatrix}=\begin{bmatrix}
                 f_\eta(x)\\f_z(x)
             \end{bmatrix}+\begin{bmatrix}
                 g_\eta(x)\\0
             \end{bmatrix}u$};
            \node[above = 0.01\linewidth of output]{\textcolor{violet}{Output+zero dynamics}};
            \draw [<->,amethyst,line width=1.5pt] (original) -- node[name=phi] {$\Phi(x)=(\eta,z)$} (output);
            \node[below = 0.01\linewidth of phi] {\textcolor{amethyst}{Diffeomorphism}};
             \node [block, below = 0.35\linewidth of output] (stable) {$\color{violet}\begin{bmatrix}\dot{\eta}\\\dot{z}\end{bmatrix}=\begin{bmatrix}
                 A_{cl}\eta\\ \tilde{f}_z(\eta)
             \end{bmatrix}$};
             \node[below = 0.01\linewidth of stable]{\textcolor{violet}{Stable dynamics}};
             \draw [->,amethyst, line width=1.5pt] (output) -- node[name=phi2] {\textcolor{amethyst}{pole }$\color{amethyst}u_\eta=K\eta$} (stable);
            \node[left = 0pt of phi2] {}; 
            \node[below right = 0pt of phi2,name=ofqp34]{\hspace{-55pt}\textcolor{amethyst}{placement}};
            \node[above right = 0pt of phi2](l) {\hspace{-55pt}\textcolor{amethyst}{linearizing}}; 
            \node[above right = 0pt of l] (f){\hspace{-55pt}\textcolor{amethyst}{Feedback}}; 
             \node [block, below = 0.4\linewidth of original] (clf) {$\color{violet}L_fV_\eta+L_gV_\eta u\leq -\lambda||\eta||^2$};
             \node[below = 0.01\linewidth of clf]{\textcolor{violet}{CLF}};
            \draw [<->,amethyst, line width=1.5pt] (clf) -- node[name=phi4] {$V=\eta^TP\eta$} (stable);
            \node[below = 0.01\linewidth of phi4,name=phi5] {\textcolor{amethyst}{Lyapunov theory}};

            \node [block, below = 0.08\linewidth of original] (qp) {$\color{violet}\min_{(u,\delta)}||u||^2+m||\delta||^2$};
            \node[ ](ppp) at (-10pt,-55pt){\textcolor{violet}{QP}}; 

            \node[above =.17\linewidth of clf,name=ofqp2]{};
            \draw [->,amethyst, line width=1.5pt] (clf) -- node[name=phi7] {\hspace{-40pt}\textcolor{amethyst}{Relaxed}} (qp);
            \node[below left = 0pt of phi7,name=ofqp3]{\hspace{-20pt}\textcolor{amethyst}{constraint}};
            \node [block2,right = 0.05\linewidth of phi7,draw=blue(pigment)] (safe) {$\color{blue(pigment)}\{x|v_{\min}\leq x\leq v_{\max}\}$};
            \node[below = 0pt of safe]{\textcolor{blue(pigment)}{Safe set} $\color{blue(pigment)}\mathcal{C}$};
            \node [block2, right= 0.07\linewidth of qp, draw=blue(pigment)] (cbf) {$\color{blue(pigment)}L_fB+L_gBu\leq \frac{\beta}{B}$};
            \node[above= 0.01\linewidth of cbf]{\textcolor{blue(pigment)}{CBF}};
            
\node [anchor=south,below = 0pt of cbf,name= safety,outer sep = 0pt,inner sep=0pt]{};
\node [name= safety2, below right= 5pt of safety] {\textcolor{amethyst}{Barrier }};
\node [name= safety3, below = 0pt of safety2] {\textcolor{amethyst}{function}};
            
            \draw [->,amethyst, line width=1.5pt, rounded corners] 
             (safe)++(42pt,0) -| node[name=phi3]{} (safety);
             \node[name=phi4,right = 0pt of phi3]{$\color{amethyst}B(x)\geq 0$};

            \node[above = 0.075\linewidth of qp,name=ofqp]{};
            \draw [->,amethyst, line width=1.5pt] (qp) -- node[name=phi8]{Control law}(original);
           \node[right = 0pt of phi8,name=phi9]{$\color{amethyst}u^*(x)$};
            \node[above = 0.075\linewidth of cbf,name=ofqp]{};
            \draw [->,amethyst, line width=1.5pt] (cbf) -- node[name=phi6] {\textcolor{amethyst}{}} (qp);
            \node[below = 0pt of phi6,name=ofqp,outer sep=0pt,inner sep=2pt]{\textcolor{amethyst}{Strict}};
            \node[outer sep=0pt,below = 0pt of ofqp,name=ofqp45,inner sep=2pt]{\textcolor{amethyst}{constraint}};

        \end{tikzpicture}
         \caption{Flow chart explaining the safety critical controller (SCC) for a single-bus DC microgrid. Note that the violet blocks on the outer loop are for stabilization, while blue blocks represent safety.}\label{fig:flow_chart}
\end{figure}
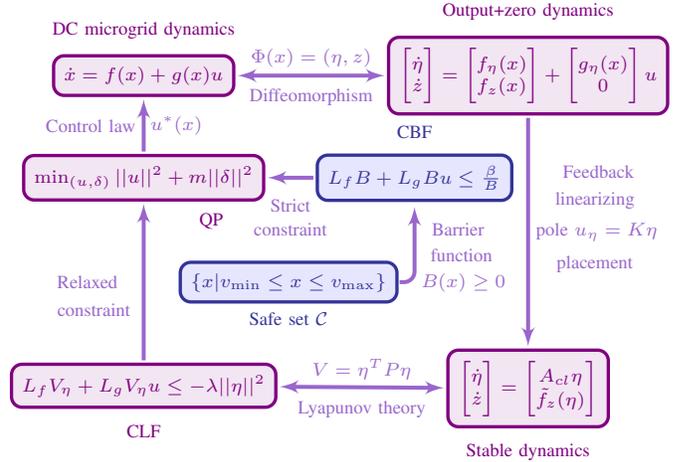

\section{Safety Critical Controller Synthesis}
Not only stability, but also the \textit{safe operation} of a DC microgrid is of utmost importance when deploying it for commercial use. The safety constraints we impose include the lower and upper bounds of the DC/DC converter output voltages, $v_j$, for $j = 1, \dots, n$. The converter terminal voltage can fluctuate significantly during contingencies such as system black startup, intermittent generation, or over/under loading. Violations of these terminal voltage constraints may lead to load shedding, generator tripping, or even cascading failures across the entire DC microgrid. In control-theoretic terms, this can be interpreted as maintaining the system state within a predefined safe region of the state space.
Similar to the CLF-based approach we used for stabilization, we will construct CBFs to ensure that the safety constraints are satisfied at all times. 
\begin{lemma}\label{lemma2}
    Suppose the safe region of the state-space is given by \begin{equation}
        \mathcal{C}:=\{x|v_{j,\min}\leq v_j\leq v_{j,\max}, j=1,\dots,n\},
    \end{equation}
    where $v_{j,\max}$ and $v_{j,\min}$ refer to strict upper and lower bounds. Then, there exist CBFs
    \begin{equation}
     B_j(x):=-\frac{1}{(v_{j}-v_{j,\min})(v_{j}-v_{j,\max})}\label{def_cbd}
    \end{equation} for $\mathcal{C}$ along the dynamics of \eqref{original}.
\end{lemma}
\begin{proof}
 We can equivalently define the safe set, $\mathcal{C}$, as the intersection of 0-super level sets of quadratic functions. Rewrite $b_j(x)=-(x_{2j-1}-x_{2j-1,\min})(x_{2j-1}-x_{2j-1,\max})$, then
$\mathcal{C}=\bigcap_{j=1,\dots,n}\{x|b_j(x)\geq0\}\label{cbf}$.

Then, $B_j(x)=\frac{1}{b_j(x)}$ functions are candidate CBFs. To confirm that they are, recall Def. \ref{CBF}, and note that $L_gB_j\neq0$ for all $x\in\text{int}(\mathcal{C})$ given that 
$v_{j,\max} \neq v_{j,\min}$. We can characterize a set of safe controllers $u$ by the solution to the following constraints 
\begin{align}
    L_fB_j+L_gB_ju\leq \frac{\beta}{B_j}, j=1,\dots,n.
\end{align} 
\end{proof}
\vspace{-15pt}
\begin{table}[hbpt]
    \centering
    \caption{DC microgrid parameters in the numerical experiment: (from left to right) output capacitance of the input converter $ i=1,\dots,5$; series inductance, followed by the series resistance of the transmission line $i=1,\dots,5$; constant power rating of the CPL connected to the bus.}
    \begin{tabular}{c|c|c|c|c}
        Node $i$& $   C_i $  [mF] & $   L_i$  [H] & $   R_i$  [m$\Omega$] & $   P_i$  [W] 
       \\
     \hline
         1&   0.49 &   0.09 &   8.78&  -\\ 
        2&   0.47 &    0.08  &    17.78 &   - \\ 
         3 &    0.49  &    0.09  &    16.78 &   - \\ 
         4 &    0.57  &    0.09  &    19.78 &   - \\ 
         5 &    0.47  &    0.08  &    27.78 &   - \\ 
       Bus &    0.47  &  -  &    1500  &    1875 
     \\
    \end{tabular}
    \label{tab:my_label1}
\end{table}
The safety-critical controller (SCC) is then synthesized by combining the CLF from Lemma 1 and CBF from Lemma 2 as constraints of a quadratic program. Visual summary of the proposed controller is given Fig. \ref{fig:flow_chart}. We are now ready to state our main result. 

\begin{theorem}\label{ours}
    If $V_\eta$ is defined as \eqref{clf_eta_clf}, and $\{B_j\}_{j=1}^n$ are defined as \eqref{def_cbd}, and $u_{FL}(x)$ as \eqref{ufl}, then the following QP
\begin{align}
            u^*(x)=&\argmin_{(u,\delta)\in\mathbb{R}^{2n}}||u-u_{FL}(x)||^2+m||\delta||^2\nonumber\\
        &\text{subject to }\nonumber\\&\hspace{0cm}\gamma(L_fV_\eta(x)+\alpha||\eta(x)||^2)+L_gV_\eta(x)\left(u+\delta\right)\leq\ 0\nonumber,\\
        &\hspace{0cm}\displaystyle L_fB_j(x)+L_gB_j(x)u\leq \frac{\beta}{B_j(x)},\nonumber\\
        &\hspace{0cm}\displaystyle j=1,\dots,n.
        \label{qp2}
\end{align}
where $\gamma(p):=\begin{cases}
    \frac{m+1}{m}p, p\geq0\\
    p, p<0,
\end{cases}$ and $m>0$ stabilizes the single-bus DC microgrid while ensuring the safety of the terminal converter voltages\eqref{original}.  
\end{theorem}
\begin{proof}
    $V_\eta$ is a CLF from Lemma \ref{lemma1} and $\{B_j\}_{j=1}^n$ are CBFs from Lemma \ref{lemma2}. Notice that each condition on $B_j$  only constrains the corresponding control input $u_j$ resulting in the CBF conditions being independent from one another. Then, applying Theorem \ref{jankovic} to each $V_\eta$ and $B_j$ pair to solve for $u^* _j$, we recover the result. 
\end{proof}
\vspace{-10pt}
\section{Numerical Experiments}
\begin{figure*}[hbpt]
\begin{multicols}{3}
    \includegraphics[width=\linewidth]{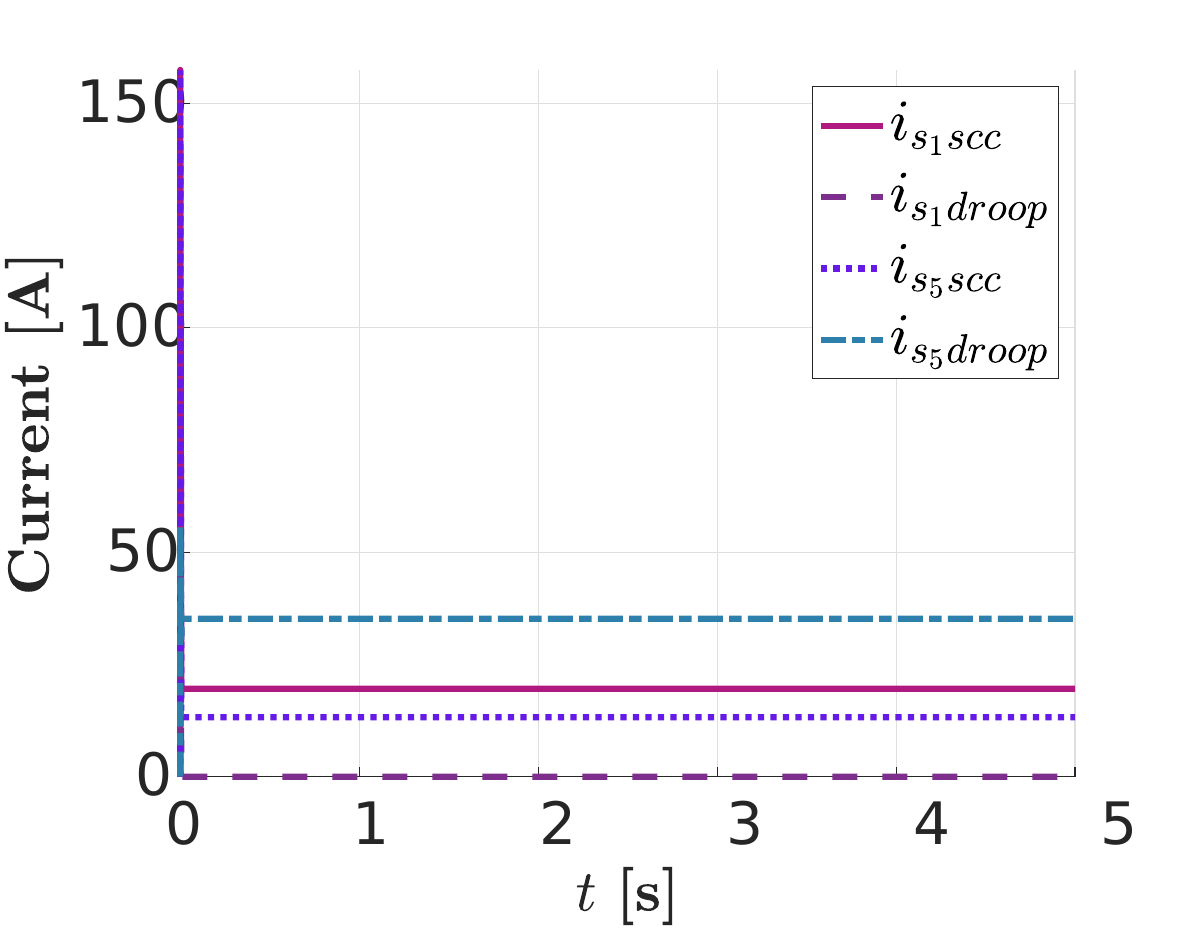}\par
    \includegraphics[width=\linewidth]{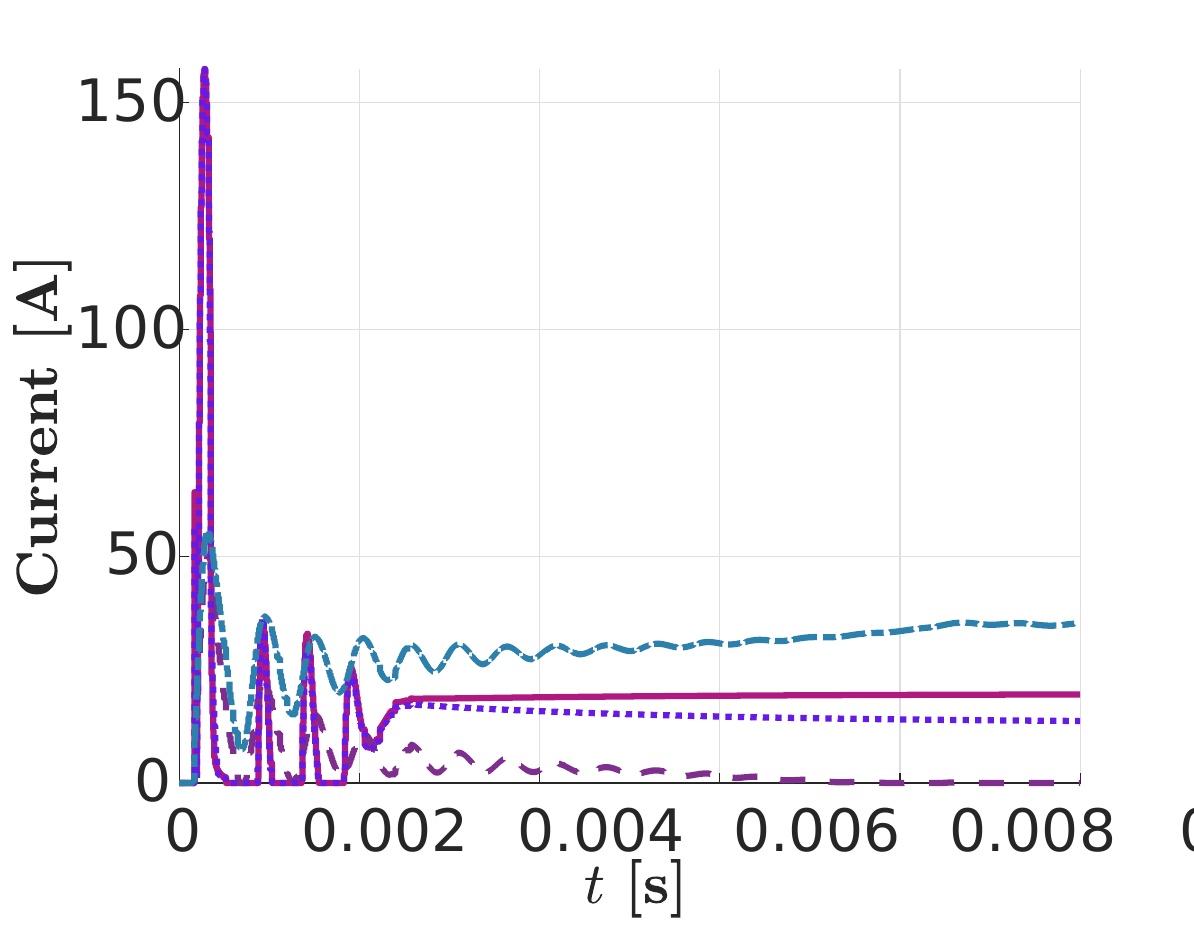}\par 
    \includegraphics[width=\linewidth]{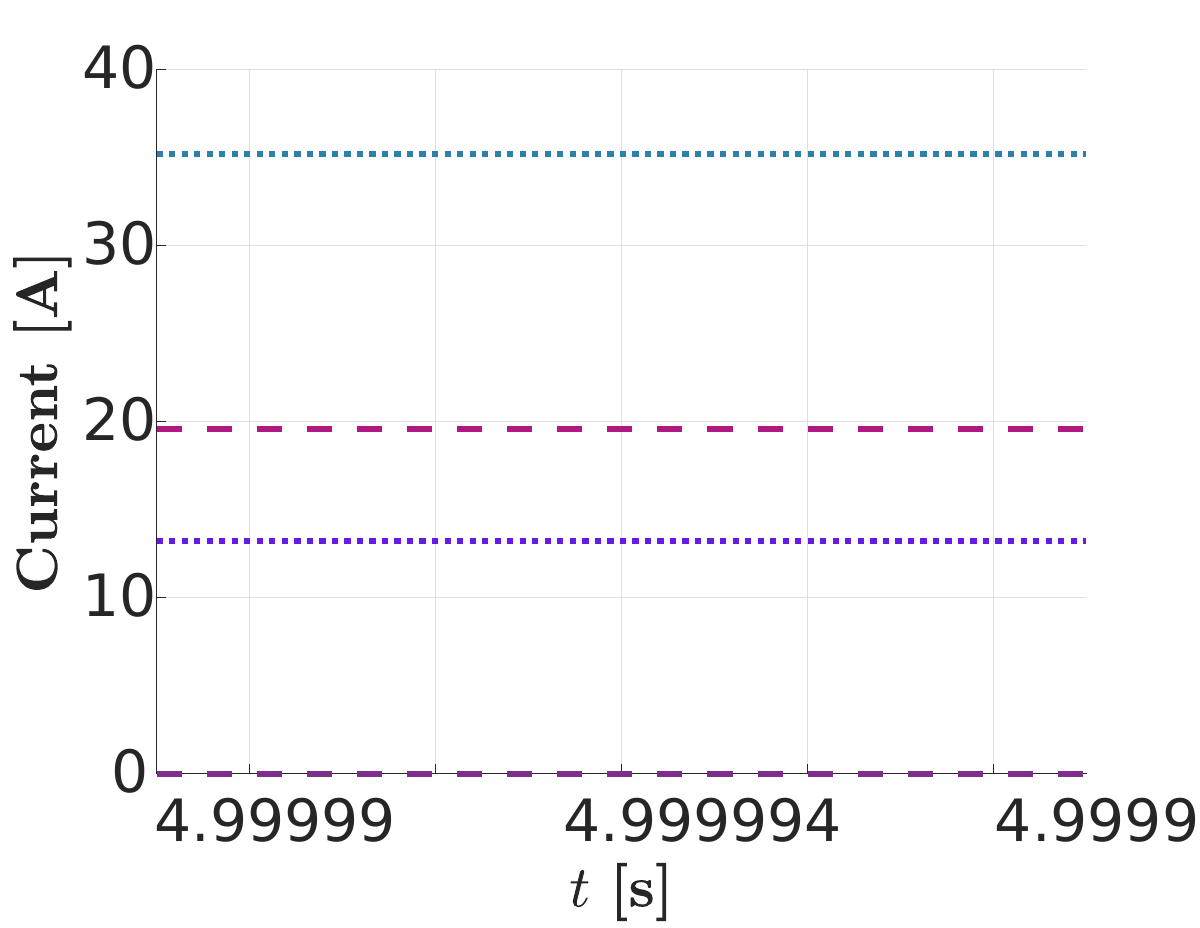}\par 
\end{multicols}
\begin{multicols}{3}
   \includegraphics[width=\linewidth]{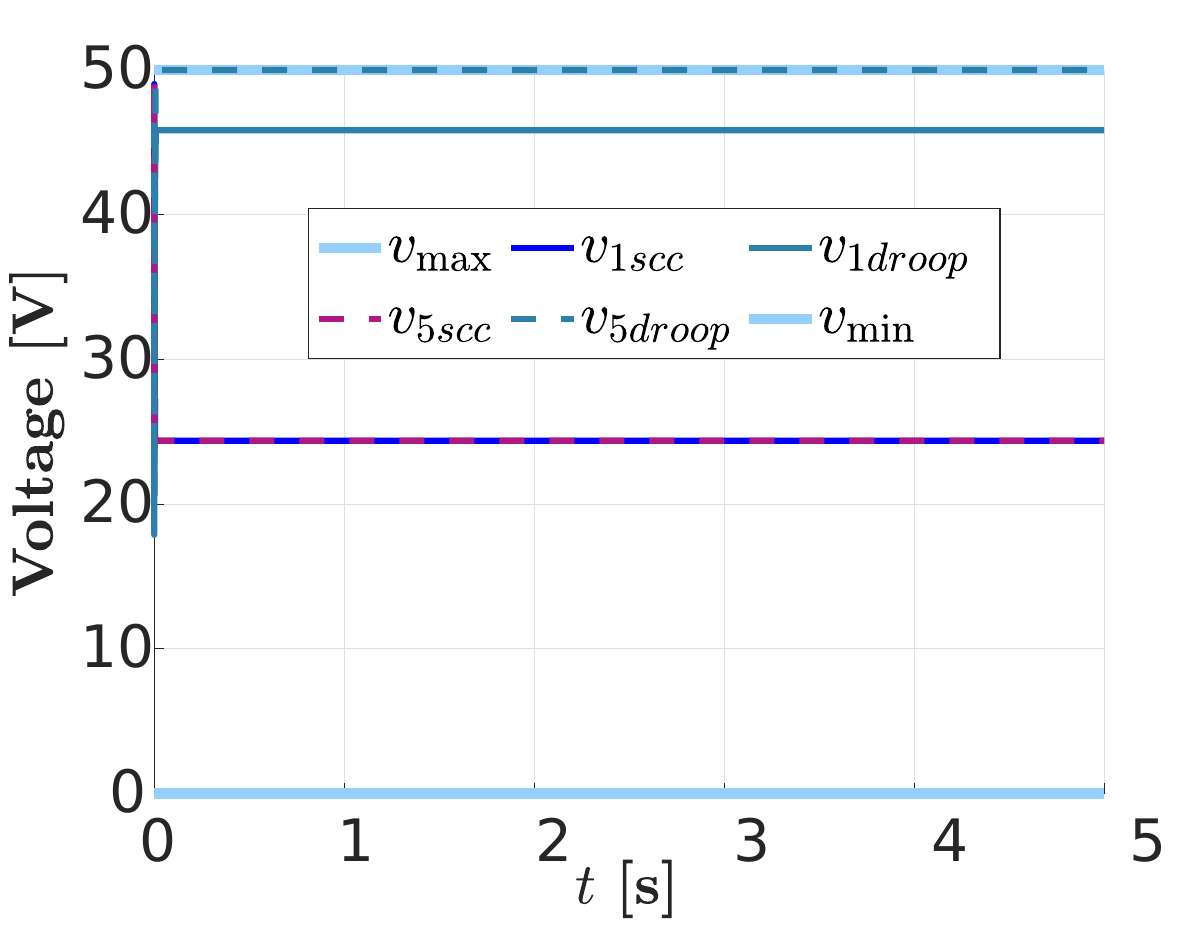}\par 
     \includegraphics[width=\linewidth]{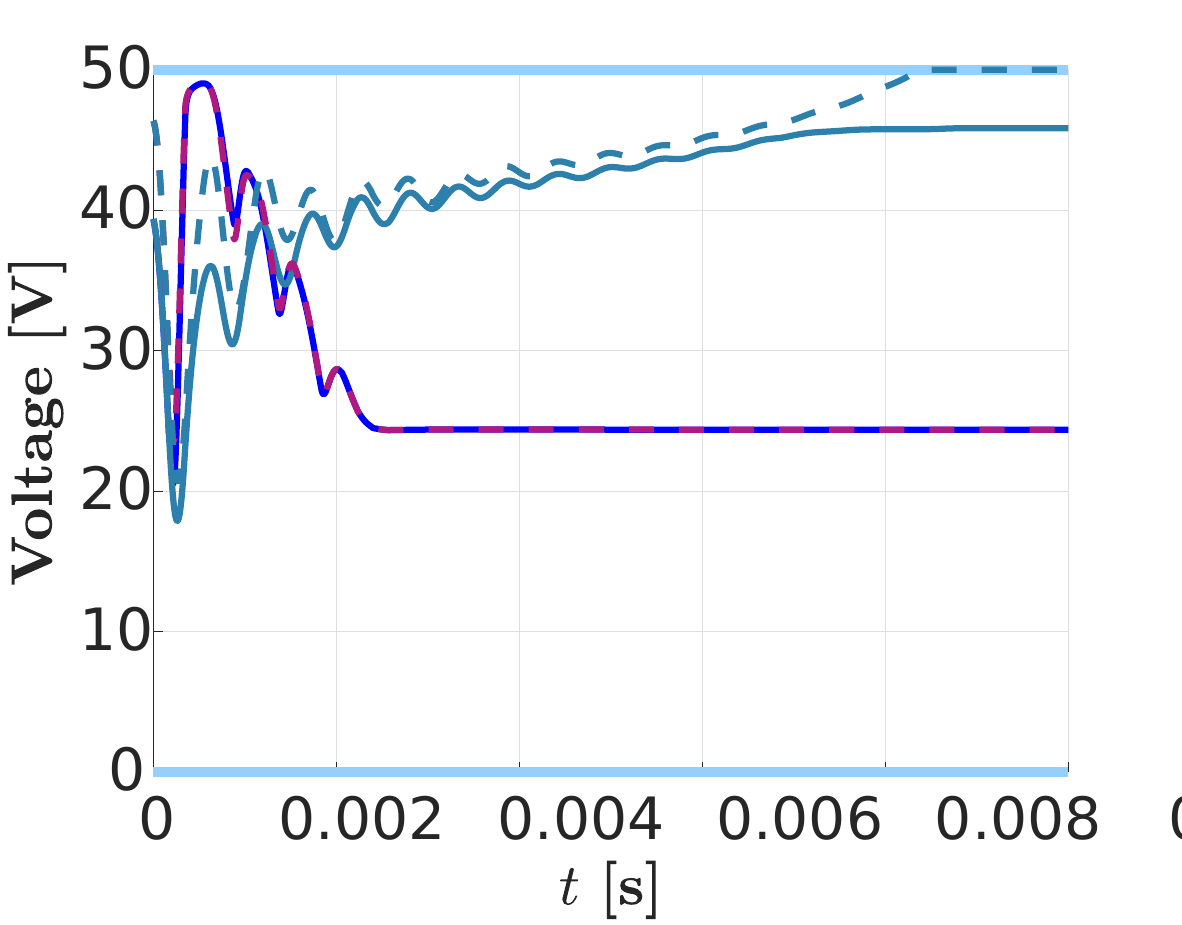}\par
    \includegraphics[width=\linewidth]{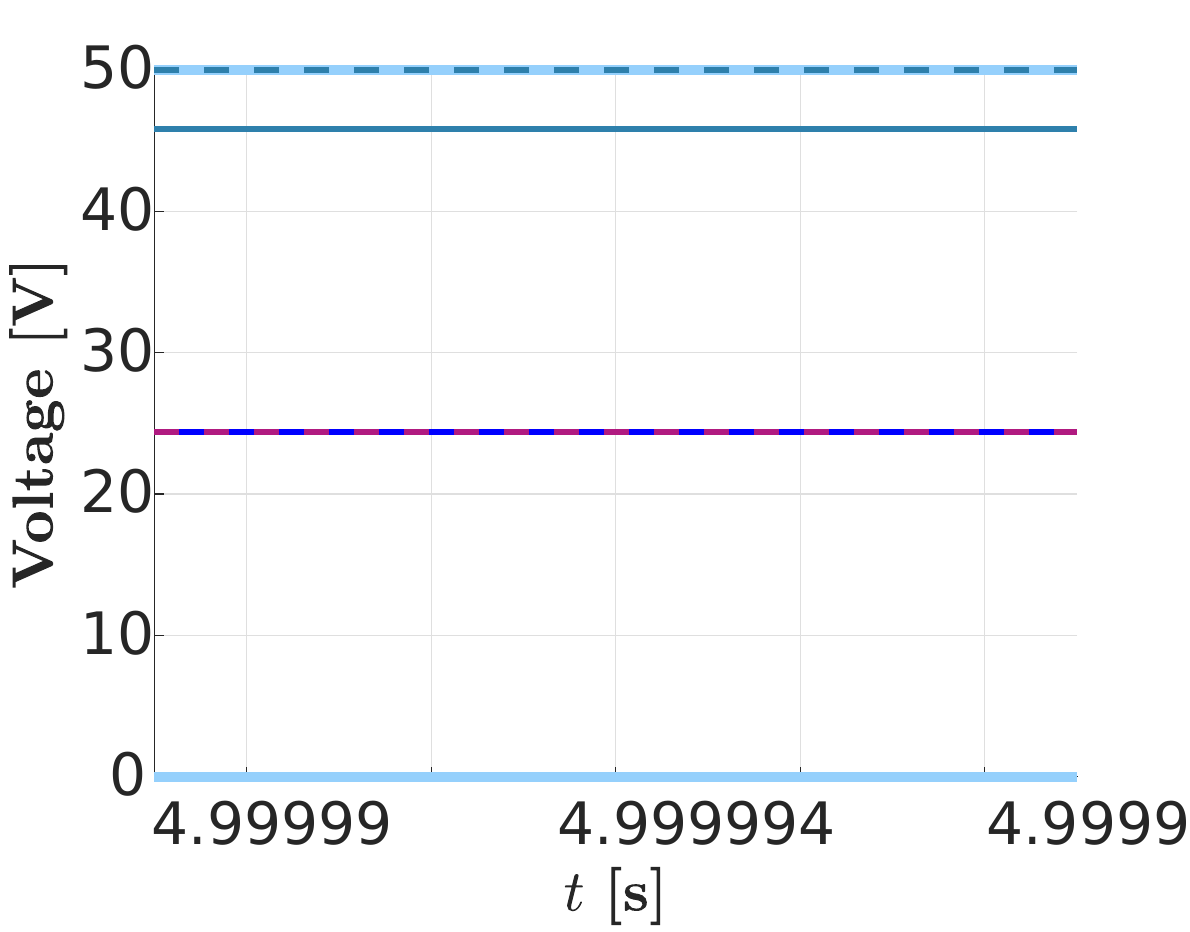}\par 
\end{multicols}
\begin{multicols}{3}
   \includegraphics[width=\linewidth]{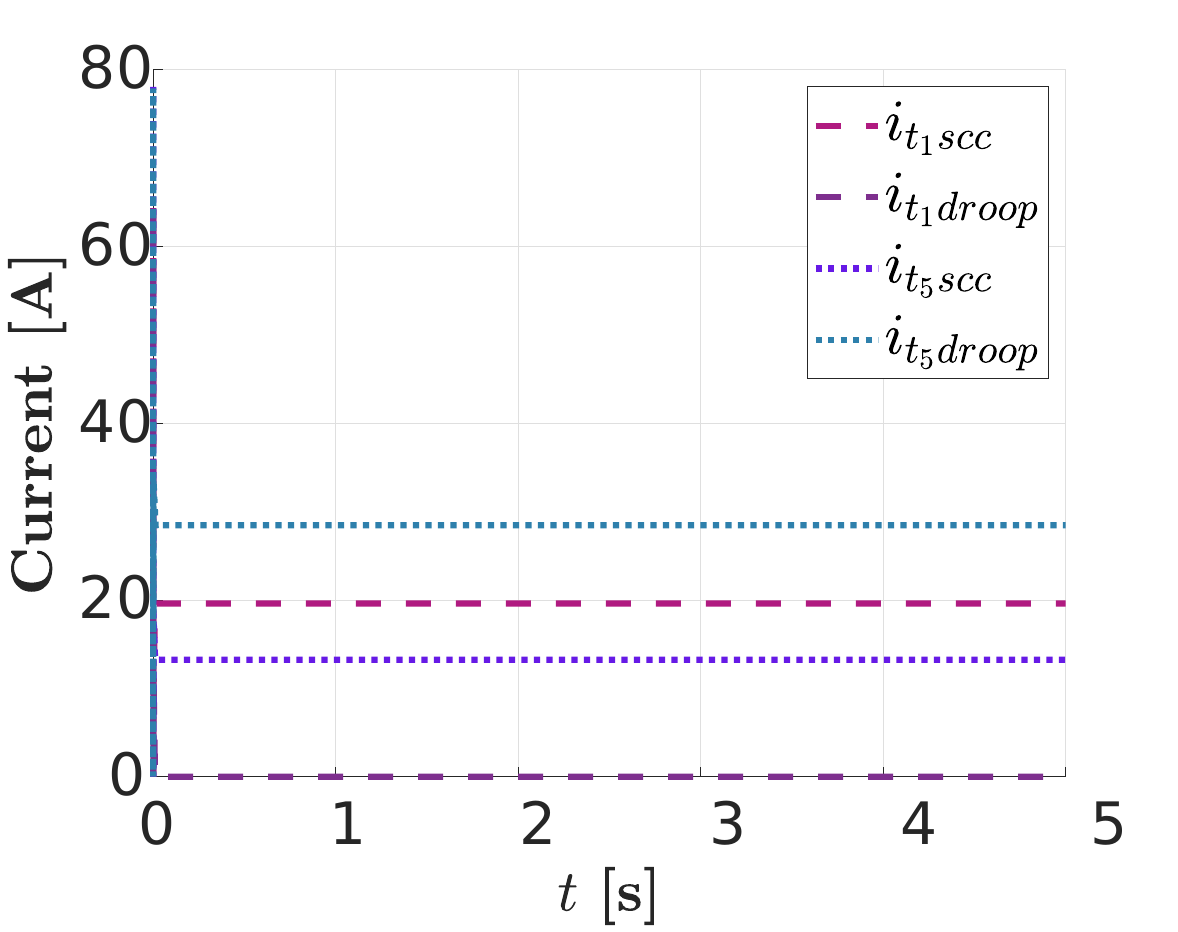}\par 
    \includegraphics[width=\linewidth]{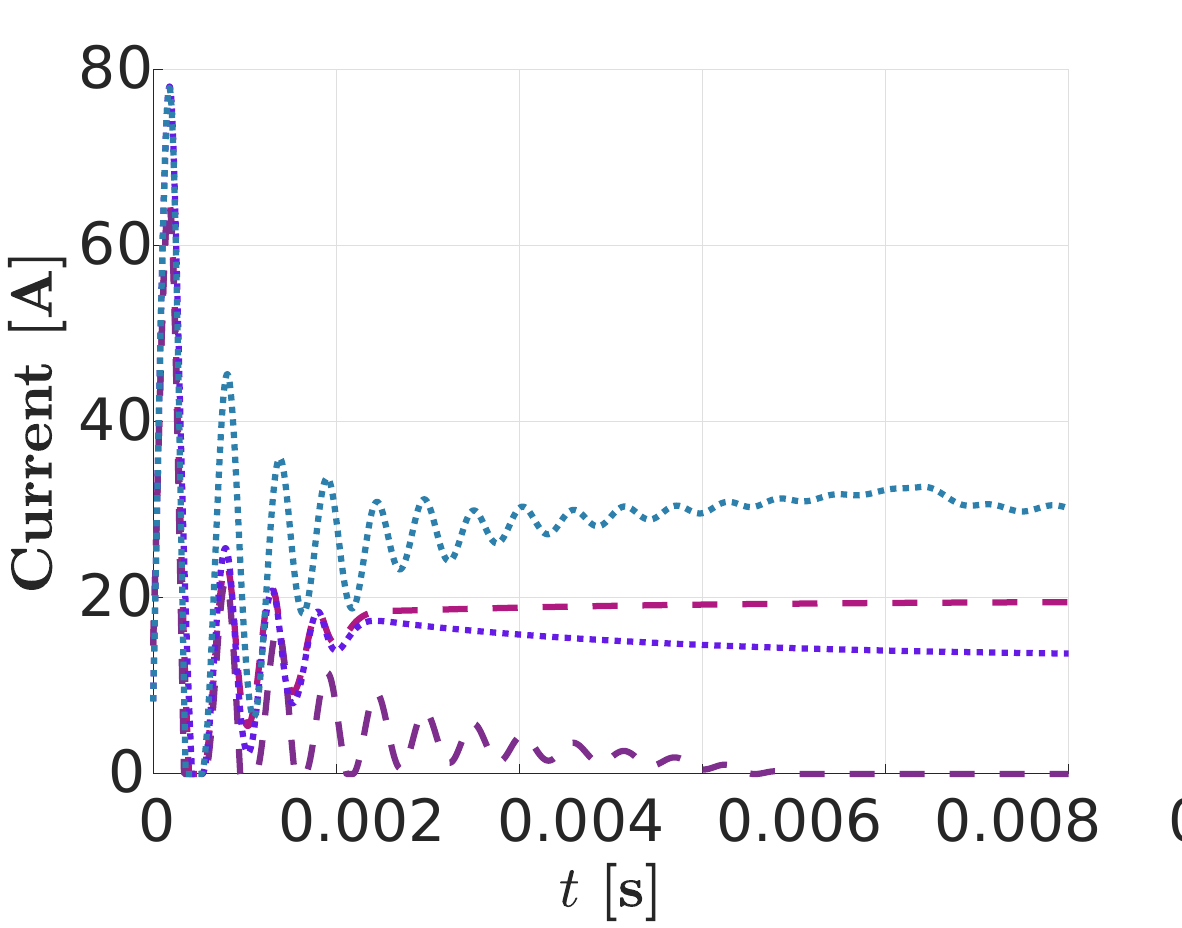}\par 
     \includegraphics[width=\linewidth]{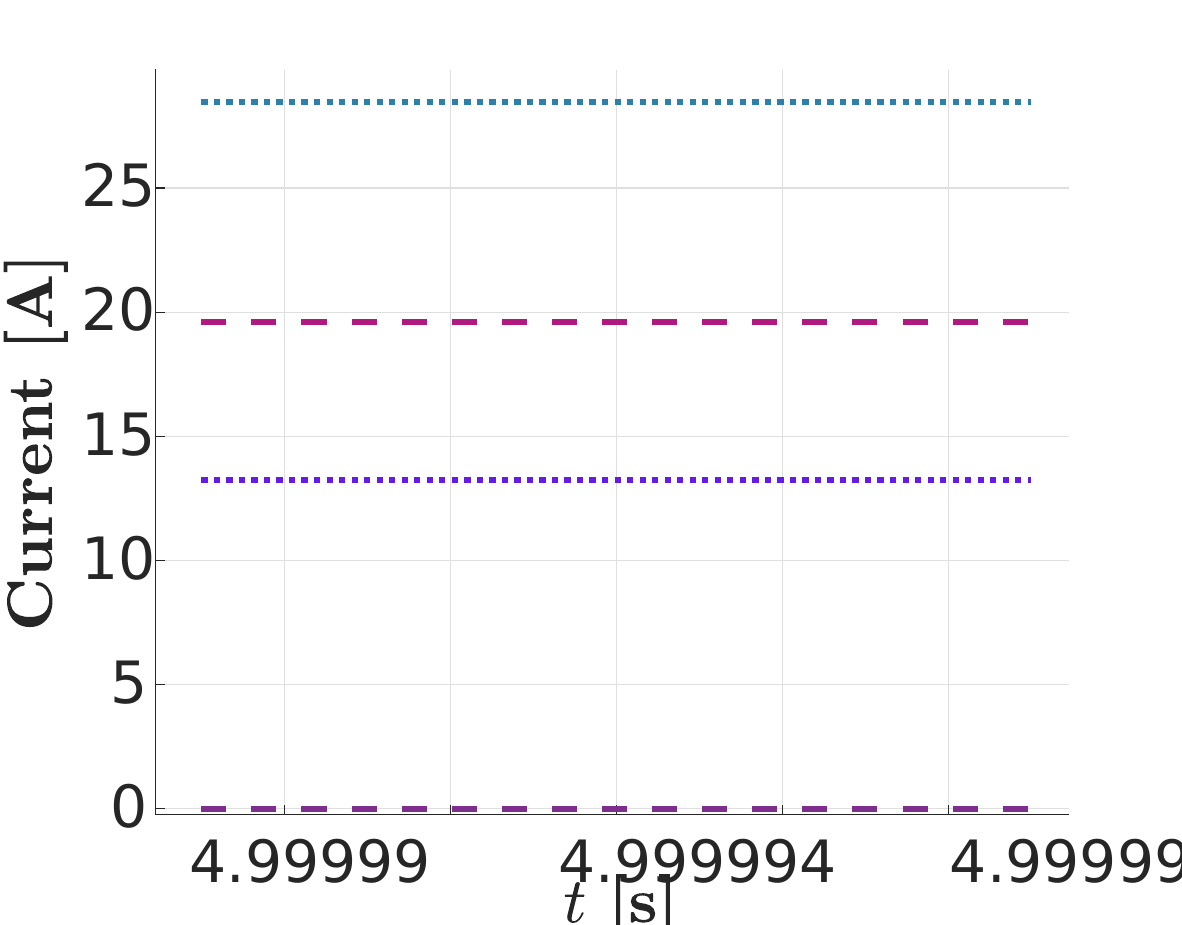}\par

\end{multicols}

\begin{multicols}{3}
    \includegraphics[width=\linewidth]{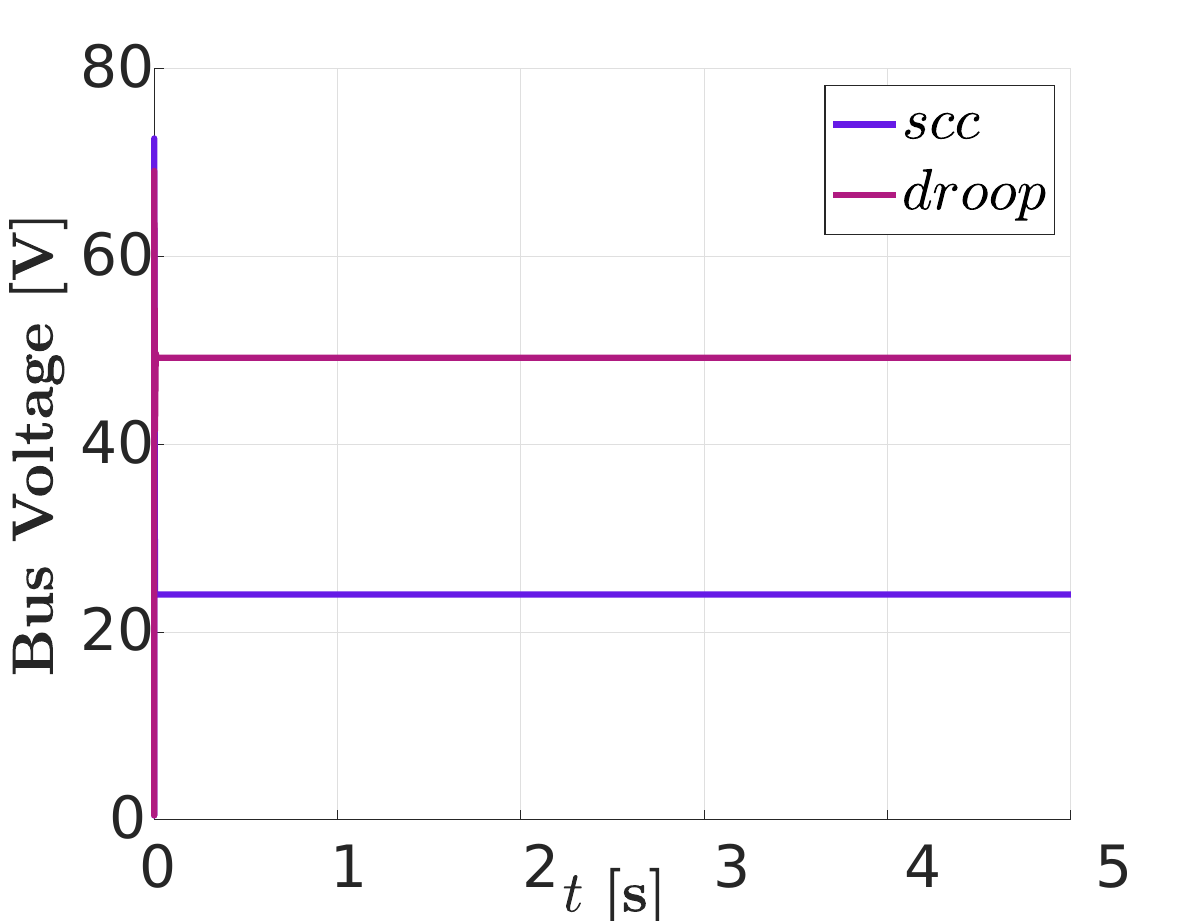}\caption*{(a) Full-length state trajectories}\label{a}\par 
     \includegraphics[width=\linewidth]{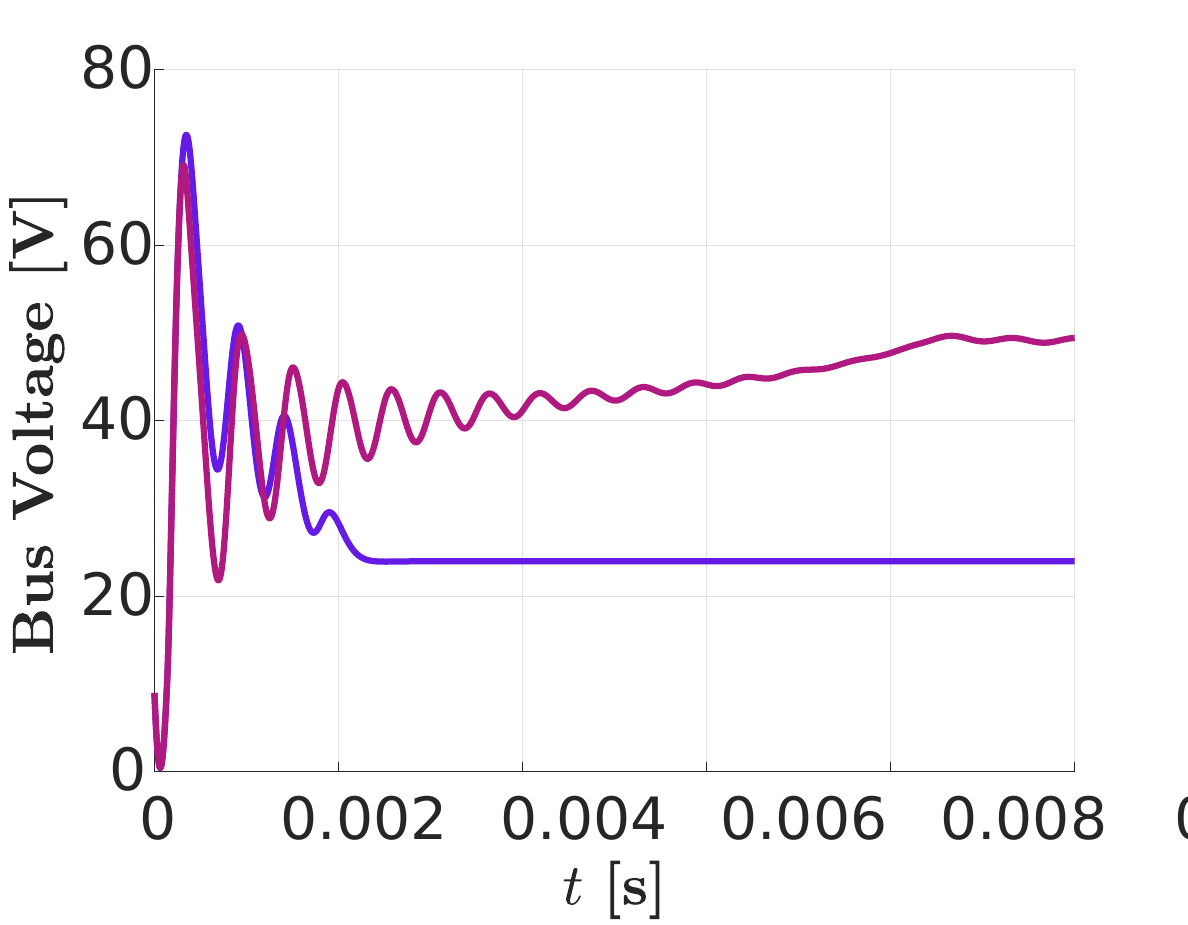}\caption*{(b) Transient behavior}\label{b}\par
    \includegraphics[width=\linewidth]{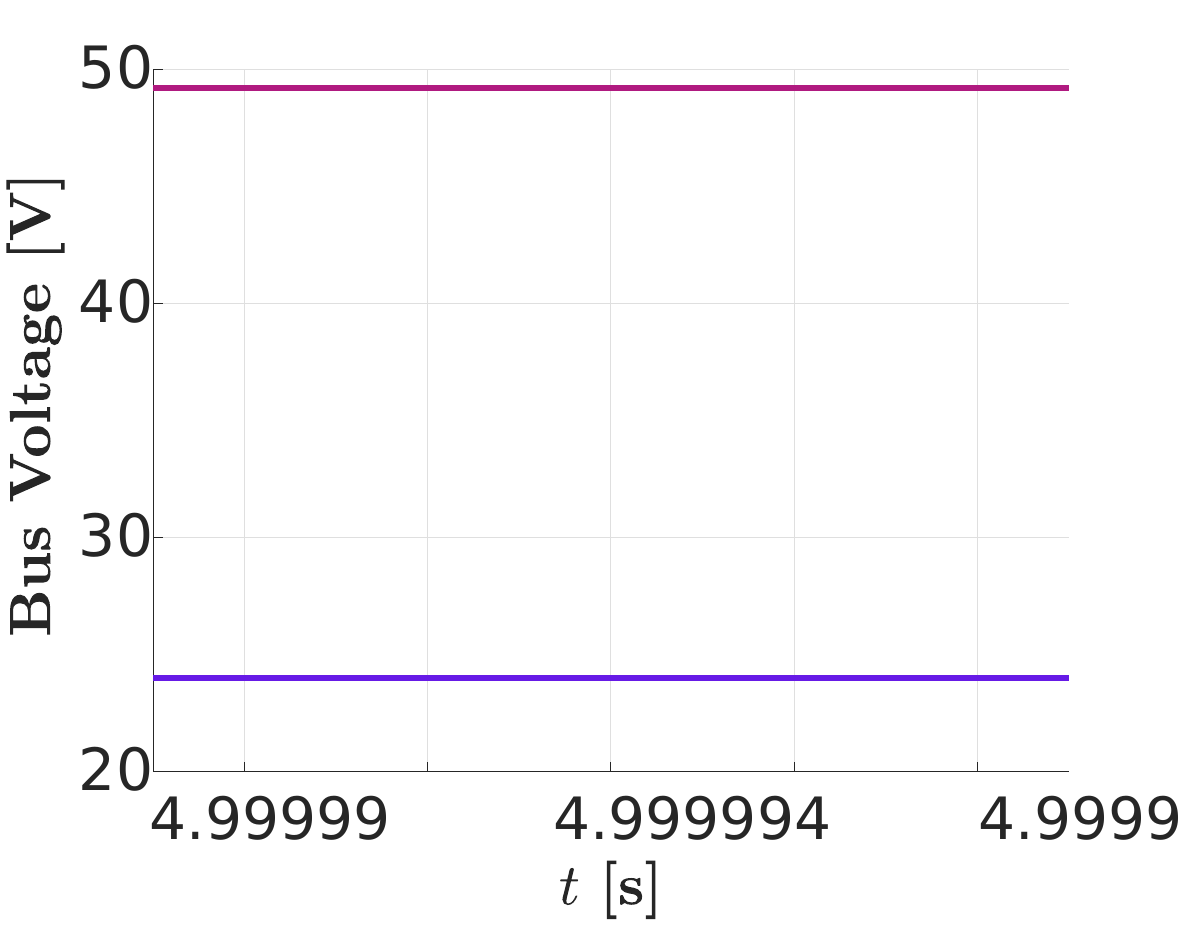}\caption*{(c) Steady-state behavior}\label{c}\par 
\end{multicols}

\caption{Comparison of the performance of the proposed SCC against the optimized robust droop controller in safe regulation of the single-bus DC microgrid to its desired steady-state bus voltage. Note that only select states and control inputs are shown for conciseness. The dynamics was updated every 1$\mu$s, while the controllers were updated every 10$\mu$s to emulate the practical implementation.}\label{fig:4}
\end{figure*}

We consider a single-bus DC microgrid comprising five source-interface DC/DC converters supplying a CPL and a resistive load in parallel. The system parameters are provided in Table \ref{tab:my_label1}. The primary control objectives are: (i) ensuring that the output voltages of the source-interface converters remain within their specified safety limits, and (ii) regulating the transmission line currents and bus voltage to their desired equilibrium values, as listed in the first two columns of Table \ref{tab:my_label}.  
The simulation results, presented in Fig. \ref{fig:4}, compare the performance of the proposed SCC against a state-of-the-art optimized droop controller \cite{herrera2015stability}, which enforces voltage constraints by saturating at the upper and lower safety bounds. The SCC's design hyperparameter $m$ is selected to be large enough to guarantee convergence to the equilibrium for the given initial conditions, but small enough so that the controller doesn't do huge jumps in value. The steady-state values achieved by the SCC and the droop controller are reported in the second-to-last and last columns of Table \ref{tab:my_label}, respectively.  
The results demonstrate that the proposed SCC effectively stabilizes the system and ensures voltage safety, even when initialized far from equilibrium. In contrast, the optimized droop controller \cite{herrera2015stability} struggles to maintain stability under certain conditions.

\section{Conclusion}
\vspace{-10pt}
This work presented an online safety-critical control (SCC) framework for single-bus DC microgrids, designed to ensure both stability and operational safety in real time. The numerical results demonstrate that the SCC outperforms a robust droop controller, achieving safer and more reliable operation across a broader set of initial conditions. The future work will focus on extending this framework to multi-bus DC microgrids, where dynamic interactions between multiple CPLs and converters present additional stability challenges.

\setlength{\tabcolsep}{4pt}
\begin{table}[t]
    \centering
    \caption{Results of the numerical experiment tabulated. The desired bus voltage is specified to be 24\,V, the remaining steady-state values are obtained from applying \eqref{u^*} to Table \ref{tab:my_label1} parameters. Initial states are picked to be far enough away from the corresponding equilibria. The steady-state values are averaged over 10$\mu$s in case there are remaining transient harmonics.}
    \begin{tabular}{c|c|c|c|c|c}
       & & & Averaged  & Averaged \\
    
      State  & Equilibrium   &  Initial  state  &   steady-state   &   steady-state & Units 
   \\ 
       &   &   &  by SCC &  by droop  
       \\
     \hline
    $v_1$ &    24.37 & 39.37& 24.37& 45.83&[V] \\ 
     
    $i_{s_1}$ &   19.61 &-&19.61& 0.00&[A]\\
    $i_{t_1}$ &  19.61 & 14.61&19.61&0.00&[A] \\\hline
   $v_2 $&    24.37 & 46.37&24.37& 49.95&[V]\\
   $i_{s_2}$ &   20.71 & -&20.71& 41.59&[A]\\ 
    $i_{t_2}$ &   20.71 & 15.71&20.71& 41.59 &[A]\\\hline 
    $v_3$   &   24.37 & 9.37&24.37&49.22&[V] \\
   $i_{s_3}$ &   21.94 & -&21.94& 0.86&[A]\\
   $i_{t_3}$ &   21.94 & 16.94&21.94&0.86&[A] \\\hline
    $v_4$  &   24.37 & 39.37&24.37&42.10&[V]  \\
   $i_{s_4}$ &   18.61 & -&18.61& 0.00&[A]\\ 
    
     $i_{t_4}$&   18.61 & 13.61&18.61& 0.00& [A]\\\hline 
     
    $v_5$  &   24.37 & 46.37&24.37& 50.00&[V]\\
   $i_{s_5}$ &   13.25 & -&13.25& 35.20&[A]\\ 
     
      $i_{t_5}$ &   13.25 & 8.25&13.25&28.46 &[A]\\ 
     \hline
    $v_{bus}$ &   24.00 & 9.00&24.00&49.21&[V]   
     \\
    \end{tabular}
    \label{tab:my_label}
\end{table}
\bibliographystyle{IEEEtran}
\bibliography{acc_main_v12}
\end{document}